\documentclass[a4paper,11pt]{article}

\usepackage{fullpage}

\usepackage{latexsym}
\usepackage{url}
\usepackage{listings}
\usepackage{xspace}
\usepackage{color}

\usepackage{amsthm}
\usepackage{amsmath}

\usepackage{algorithm}
\usepackage{algpseudocode}

\usepackage{booktabs} 

\usepackage{tikz}

\newcommand{\eg}{e.g.\@\xspace}

\newcommand{\band}{\wedge}
\newcommand{\isodd}{\mathrm{odd}}
\newcommand{\isevn}{\mathrm{even}}
\newcommand{\shifthi}[2]{#1\!\uparrow\!#2}
\newcommand{\shiftlo}[2]{#1\!\downarrow\!#2}
\newcommand{\popcount}{\mathtt{popcnt}}

\newcommand{\ceil}[1]{\lceil #1\rceil}

\newcommand{\mpisendrecv}{\textsf{MPI\_\-Sendrecv}\xspace}

\newcommand{\mpibarrier}{\textsf{MPI\_\-Barrier}\xspace}
\newcommand{\mpiscan}{\textsf{MPI\_\-Scan}\xspace}
\newcommand{\mpiexscan}{\textsf{MPI\_\-Exscan}\xspace}
\newcommand{\mpiallreduce}{\textsf{MPI\_\-Allreduce}\xspace}
\newcommand{\mpireducelocal}{\textsf{MPI\_\-Reduce\_local}\xspace}

\newcommand{\mpilong}{\texttt{MPI\_LONG}\xspace}
\newcommand{\mpibxor}{\texttt{MPI\_BXOR}\xspace}

\newcommand{\senddata}{\textsf{Send}\xspace}
\newcommand{\recvdata}{\textsf{Recv}\xspace}

\newcommand{\bidirec}[2]{\textsf{Send}(#1)\parallel\textsf{Recv}(#2)\xspace}

\newtheorem{proposition}{Proposition}

\newtheorem{corollary}{Corollary}

\newcommand{\gcc}{\texttt{gcc~12.1.0}\xspace}
\newcommand{\hydraopenmpi}{OpenMPI~4.1.6\xspace}

\newcommand{\hydrampich}{\texttt{mpich~4.1.2}\xspace}

\newcommand{\benchrep}{$155$\xspace}
\newcommand{\benchwarm}{$12$\xspace}

\title{Two Efficient Message-passing Exclusive Scan Algorithms\thanks{This note improves significantly over \texttt{arXiv:2507.04785}}}

\author{Jesper Larsson Tr\"aff\\
  TU Wien\\
  Faculty of Informatics\\
  Institute of Computer Engineering, Research Group Parallel Computing 191-4\\
  Treitlstrasse 3, 5th Floor, 1040 Vienna, Austria}
\date{April 2026}

\begin{document}

\maketitle

\begin{abstract}
  Parallel scan primitives compute element-wise inclusive or exclusive
  prefix sums of input vectors contributed by $p$ consecutively ranked
  processors under an associative, possibly expensive, binary operator
  $\oplus$. In message-passing systems with bounded, one-ported
  communication capabilities, at least $\ceil{\log_2 p}$ or
  $\ceil{\log_2 (p-1)}$ send-receive communication rounds are required
  to perform the scans.  While there are well-known, simple algorithms
  for the inclusive scan that solve the problem in $\ceil{\log_2 p}$
  send-receive communication rounds with $\ceil{\log_2 p}$
  applications of the $\oplus$ operator, the exclusive scan is
  different and has been much less addressed.

  By considering natural invariants for the exclusive prefix sums
  problem, we present two different algorithms that are efficient in
  the number of communication rounds and in the number of applications
  of the $\oplus$ operator. The first algorithm consists of an
  inclusive scan phase and an exclusive scan phase and trades the
  number of communication rounds against the number of applications of
  the $\oplus$ operator.  For any chosen number of inclusive scan
  communication rounds $q'\geq 1$, the algorithm performs
  $q=\ceil{\log_2(p-1)+\log_2\frac{2^{q'}}{2^{q'}-1}}=
  \ceil{\log_2\frac{2^{q'}}{2^{q'}-1}(p-1)}$ send-receive
  communication rounds with $q+q'-2$ applications of the $\oplus$
  operator. The smallest number of inclusive scan rounds with
  $q=\ceil{\log_2 p}$ rounds in total is $q'\geq
  q-\log_2(2^q-p+1)$. It includes as special cases two standard
  implementations of the exclusive scan operation, namely reducing the
  problem to an inclusive scan problem on $p-1$ processors at the cost
  of one extra communication round, or modifying an inclusive scan
  algorithm to also compute the exclusive scan at the cost of twice as
  many operator applications. The other algorithm is a modification of
  a round-optimal all-reduce algorithm, and the number of additional
  applications of the $\oplus$ operator is dependent on the number of
  bits set (popcount of) in $p-1$.  Both algorithms are relevant for
  small(er) input vectors where performance is dominated by the number
  of communication rounds. For large input vectors, other (pipelined,
  fixed-degree tree) algorithms must be used.

  Parallel scan primitives are included in MPI as the \mpiscan and
  \mpiexscan collectives.  An experimental comparison with
  implementations in MPI of the \mpiexscan collective operation
  indicates that the considerations here are of practical relevance.
\end{abstract}


\section{Introduction}

The parallel scan is a fundamental primitive in parallel
algorithmics, typically used for bookkeeping and load balancing
purposes, but often also directly in and for the algorithms
themselves, see for
instance~\cite{Blelloch89,LakhotiaPetriniKannanPrasanna22,CopikGrosserHoeflerBientinesiBerkels22}.
Scan operations compute the inclusive or exclusive prefix sums over
inputs provided by $p$ successively ranked (numbered) processors under
a given, associative, binary operator $\oplus$. Each of the $p$
processors $r,0\leq r<p$ has an input vector $V_r$ of $m$ elements and
computes its element-wise, $r$th prefix sum $W_r$ as follows:
\begin{itemize}
\item
  Inclusive scan, for $0\leq r<p$:
  \begin{align*}
    W_r & = \oplus_{i=0}^{r}V_i
  \end{align*}
\item
  Exclusive scan, for $0<r<p$:
  \begin{align*}
    W_r & = \oplus_{i=0}^{r-1}V_i
  \end{align*}
\end{itemize}

For distributed memory, parallel message-passing programming, both
inclusive and exclusive scan operations are commonly standardized in
this way, \eg, in MPI~\cite{MPI-4.1} as \mpiscan and \mpiexscan,
respectively. Both primitives find extensive
use~\cite{ChunduriParkerBalajiHarmsKumaran18,LagunaMarshallMohrorRufenachtSkjellumSultana19}.

Arguably, the exclusive scan is the more important, more often used
variant of the scan primitives~\cite{LakhotiaPetriniKannanPrasanna22}.
In a shared-memory setting, where input and output are stored in
shared-memory arrays, the difference is possibly not too important,
since the $r$th exclusive prefix equals exactly the $(r-1)$th
inclusive prefix which can be read immediately in the result array. In
a distributed memory, message-passing setting, this reduction from
exclusive to inclusive scan involves additional communication from processor
$r-1$ to processor $r$ for all processors $r,0<r<p$.

Nevertheless, message-passing algorithms in the literature deal mostly
with the inclusive scan primitive, and trivially reduce the exclusive
scan primitive to an inclusive scan with one extra communication round
either after or before the inclusive scan (which need to be performed
on only $p-1$ processors).  As can easily be seen and as has been
discovered many times, the inclusive scan problem can be solved in
$\ceil{\log_2 p}$ communication rounds in which processors
simultaneously send and receive partial
results~\cite{HillisSteele86,KoggeStone73,KruskalRudolphSnir85} (or
$\ceil{\log_{k+1} p}$ rounds for $k$-ported systems~\cite{LinYeh99}):
We will recapitulate the algorithm in Section~\ref{sec:algorithm}.  In
systems with one-ported communication capabilities, this is optimal,
since the last processor $r=p-1$ needs information (partial results)
from all $p$ processors and the number of processors from which $r$
can have information after $k$ communication rounds is at most
$2^k$. These algorithms send and receive full, $m$-element input or
partial result vectors in each communication round, and are therefore
mostly relevant for vectors with small numbers of elements $m$.  For
large vectors, pipelined, fixed-degree tree algorithms should be used
\cite{LakhotiaPetriniKannanPrasanna22,Traff09:twotree,Traff06:scan},
which, however, require a larger number of communication rounds and
are therefore not as suitable for small vectors. MPI libraries use
different combinations of algorithms and different choices for their
implementations of \mpiscan and \mpiexscan, but, as can be seen in
Appendix~\ref{sec:evaluation}, often select logarithmic round,
full-vector algorithms up to quite large values of $m$.

In this paper, we examine algorithms for the exclusive scan
operation, focusing on the number of communication rounds which is one
factor influencing the achievable and observable performance for small
vectors on real, distributed-memory parallel computing systems.  A
round lower bound is clearly $\ceil{\log_2 (p-1)}$ since a reduction
of $p-1$ inputs has to be computed; whether this is tight seems not to
be known.  The reduction of the exclusive scan to an inclusive scan
plus an extra communication step will in the best case take
$1+\ceil{\log_2 (p-1)}$ communication rounds, which is at least one
more than the lower bound. Modifying an inclusive scan algorithm to
compute, in essence, both inclusive and exclusive scans, leads to an
algorithm taking $\ceil{\log_2 p}$ communication rounds, but at the
expense of two applications of the $\oplus$ operator per round (except
for the first round, so the number of applications is actually
$2\ceil{\log_2 p}-2$). This can be a significant drawback when
$\oplus$ is expensive and as $m$ becomes larger. Performing an
inclusive scan and then letting each processor remove its own
contribution would give $\ceil{\log_2 p}$ communication rounds with
$\ceil{\log_2 p}+1$ operator applications, but works only if the
$\oplus$ operator has an inverse $\oplus^{-1}$ which is often not the
case (min/max operators).

We first present an algorithm that combines the two observations and
performs the exclusive scan in $q=\ceil{\log_2
  (p-1)+\log_2\frac{2^{q'}}{2^{q'}-1}}$ simultaneous send-receive
communication rounds with $q-q'-2$ applications of the $\oplus$
operator where $q'$ is the number of initial inclusive scan
communication rounds and can be chosen freely. A smallest $q'$ for
which the total number of rounds is $q=\ceil{\log_2 p}$ is given by
$q'\geq q-\log_2(2^q-p+1)$.

Secondly, we modify a round-optimal algorithm for the all-reduce
(\mpiallreduce)
operation~\cite{BarNoyKipnisSchieber93,Traff25:reducescatter} to solve
instead the exclusive prefix sums problem. This algorithm takes always
$\ceil{\log_2 p}$ communication rounds, but the number of additional
applications of the $\oplus$ operator depends on the popcount (number
of bits set) in $p-1$ as $\popcount(p-1)-1-\isevn(p)$.

It remains open whether an algorithm running in $\ceil{\log_2 (p-1)}$
rounds with at most $\ceil{\log_2 (p-1)}$ applications of $\oplus$
exists.

In the appendix, we collect experimental evidence on a small
high-performance cluster with $36$ compute nodes, each with $32$
processor cores with two commonly used, state-of-the-art MPI libraries to
show that the considerations, apart from being theoretically interesting,
can be worthwhile and have concrete, practical impact for MPI implementations
of the parallel exclusive scan primitive.

\section{Inclusive and Exclusive Scan Algorithms}
\label{sec:algorithm}

\begin{algorithm}
  \caption{Algorithm for the exclusive prefix sums operation for
    processor $r$ with adjusted doubling skips $s$. Each processor has
    input in $V$ and computes the $r$th element wise exclusive prefix
    sum in rank order into $W$. The associative reduction operator is
    $\oplus$. For a given $p', p'<p$, exclusive and inclusive prefix sums are
    computed by doubling $s$ as long as $s<p'$. After this, $s$ is
    adjusted to $s-1$ and doubled further to compute only exclusive prefix
    sums as long as $s<p$. For a chosen number of inclusive scan rounds $q'$,
    set $p'=2^{q'}$. The best $p'$ is $2^\ceil{\log_2 p}-p+1$.}
  \label{alg:exscan}
  \begin{algorithmic}
    \Procedure{ExScan}{$V,W,\oplus$}
    \State $s\gets 1$ \Comment First skip
    \State $t,f\gets r+s,r-s$ \Comment To- and from-processors for first round
    \If{$0\leq f\band t<p$}
    \State $\bidirec{V,t}{W,f}$
    \ElsIf{$0\leq f$}
    $\recvdata(W,f)$
    \ElsIf{$t<p$} $\senddata(V,t)$
    \EndIf
    \State $s\gets \shifthi{s}{1}$ \Comment Double
    \While{$s<p'$} \Comment Inclusive scan phase, straight doubling 
    \State $t,f\gets r+s,r-s$ \Comment To- and from-processors
    \If{$0\leq f\band t<p$}
    \State $W'\gets W\oplus V$
    \State $\bidirec{W',t}{T,f}$
    \State $W\gets T\oplus W$
    \ElsIf{$0\leq f$}
    \State $\recvdata(T,f)$
    \State $W\gets T\oplus W$
    \ElsIf{$0<r\band t<p$}
    \State $W'\gets W\oplus V$ \Comment Do this $W'$ computation at most once
    \State $\senddata(W',t)$
    \ElsIf{$0=r\band t<p$}
    $\senddata(V,t)$
    \EndIf
    \State $s\gets \shifthi{s}{1}$
    \EndWhile 
    \State $s\gets s-1$ \Comment Adjust skip and switch to exclusive scan phase, processor $r=0$ is now done
    \While{$s<p-1$}
    \State $t,f\gets r+s,r-s$ \Comment To- and from-processors
    \If{$1\leq f\band t<p$}
    \State $\bidirec{W,t}{T,f}$
    \State $W\gets T\oplus W$
    \ElsIf{$1\leq f$}
    \State $\recvdata(T,f)$
    \State $W\gets T\oplus W$
    \ElsIf{$1\leq r\band t<p$}
    $\senddata(W,t)$
    \EndIf
    \State $s\gets \shifthi{s}{1}$ \Comment Double
    \EndWhile
    \EndProcedure
  \end{algorithmic}
\end{algorithm}

We consider algorithms where the $p$ processors communicate using a
\emph{circulant graph}
pattern~\cite{Bruck97,Traff25:reducescatter}. In a communication
round, processor $r, 0\leq r<p$ simultaneously receives a vector from
processor $r-s$ (if $\geq 0$) and sends a vector to processor $r+s$
(if $<p$) for some positive \emph{skip} $s$ which in each round is the
same for all processors. We denote the skip for round $k$ as
$s_k$. The number of processors $p$ and the sequence of some $q$ skips
$s_k, k=0,1,\ldots, q-1$, both of which depend on $p$, define the circulant
graph.  In well-known, standard, logarithmic round inclusive scan
algorithms, call them either Hillis-Steele~\cite{HillisSteele86},
Kogge-Stone~\cite{KoggeStone73} or
Kruskal-Rudolph-Snir~\cite{KruskalRudolphSnir85}, the processors
maintain a partial result of the form
\begin{align*}
  W_r & = \oplus_{i=\max(0,r-s_k+1)}^{r}V_i
\end{align*}
as an invariant that holds before round $k, k=0,1,\ldots,\ceil{\log_2
  p}-1$ for a sequence of \emph{straight doubling} skips $s_k=2^k$.
Before the first round with $k=0$, this invariant can easily be
established by locally copying $V_r$ into $W_r$, and if the invariant
holds before round $k$, processor $r$ can receive the computed partial
result $W_{r-s_k}$ from processor $r-s_k$ (as long as $r-s_k\geq 0$)
and add it with $\oplus$ to its own computed partial result
$W_r$. Since $s_k+s_k=2^k+2^k=2^{k+1}=s_{k+1}$, the invariant now
holds before round $k+1$. When $s_k>r$, the invariant implies that
processor $r$ has indeed computed the $r$th inclusive prefix sum. To
avoid unfairly giving preference, we term this algorithm the
\emph{straight doubling inclusive scan} algorithm.

The straight doubling inclusive scan algorithm can be extended compute
both the inclusive and the exclusive scan by maintaining after the
first communication round instead the invariant
\begin{align*}
  W_r & = \oplus_{i=\max(0,r-s_k+1)}^{r-1}V_i
\end{align*}
which implies that processor $r$ has computed the $r$th exclusive
prefix sum upon termination.  Processor $r$ would then have to send in
$W_r\oplus V_r$ to processor $r+s_k$ in round $k$ (as long as
$r+s_k<p$), requiring two applications of $\oplus$ in each round,
except the first.  We call this exclusive scan algorithm the
\emph{two-$\oplus$ doubling exclusive scan} algorithm.

The exclusive scan seems more elusive than the inclusive scan. The
natural invariant, complementary to the inclusive scan invariant,
would be to maintain
\begin{align*}
  W_r & = \oplus_{i=\max(0,r-s_k)}^{r-1}V_i
\end{align*}
and use a similar, doubling communication pattern. Even with $s_0=1$,
the invariant does not hold before the first round, since processor
$r$ is missing the input from processor $r-1$ (for $r>0$). In order to
establish the invariant, we first let each processor $r$ receive
$V_{r-1}$ into $W_r$ from processor $r-1$.  The invariant now holds
for the next round, and with $s_k=2^{k-1}$ for
$k=1,2,\ldots,\ceil{\log_2(p-1)}$, we can let each processor receive
the computed partial result $W_{r-s_k}$ from processor $r-s_k$ and $\oplus$
it into the partial result $W_r$. Since
\begin{align*}
  W_{r-s_k} \oplus W_r & = 
  \left(\oplus_{i=\max(0,r-s_k-s_k)}^{r-s_k-1}V_i\right) \oplus
  \left(\oplus_{i=\max(0,r-s_k)}^{r-1}V_i\right) \\
  & = \oplus_{i=\max(0,r-s_{k+1})}^{r-1}V_i
\end{align*}
the invariant is reestablished and holds again before round $k+1$. The
invariant is void for processor $r=0$, therefore after the initial
round with $s_0=1$, there is no further communication with processor
$r=0$ necessary, and $\ceil{\log_2 (p-1)}$ additional rounds suffice
to complete the exclusive scan. In total, the number of communication
rounds is $1+\ceil{\log_2 (p-1)}$. We term this exclusive scan
algorithm the \emph{$1$-doubling exclusive scan} algorithm. It is
essentially the same algorithm as first shifting the input from
processor $r$ to processor $r+1$ and then performing a straight
doubling scan on the $p-1$ processors $r,r>0$. This algorithm has the
advantage that it performs only $\ceil{\log_2 (p-1)}$ applications of
$\oplus$.

To get the number of communication rounds closer to only $\ceil{\log_2
  (p-1)}$ without doubling the number of applications of $\oplus$, a
new idea is needed. Assume we first perform $q'$ initial rounds of
the two-$\oplus$ doubling exclusive scan algorithm. In the
last of these rounds, $q'-1$, processor $r$ computes
\begin{align*}
  \left(W_{r-s_{q'-1}}\oplus V_{r-s_{q'-1}}\right) \oplus W_r & =
  \left(\oplus_{i=\max(0,r-s_{q'-1}-s_{q'-1}+1)}^{r-s_{q'-1}-1}V_i\right) \oplus
  V_{r-s_{q'-1}} \oplus
  \left(\oplus_{i=\max(0,r-s_{q'-1}+1)}^{r-1}V_i\right) \\
  &= \left(\oplus_{i=\max(0,r-s_{q'-1}-s_{q'-1}+1)}^{r-1}V_i\right) \\
  &= \left(\oplus_{i=\max(0,r-s_{q'}}^{r-1}V_i\right) \quad .
\end{align*}
The last computed $\oplus$ sum follows if we take
$s_{q'}=s_{q'-1}+s_{q'-1}-1=2^{q'}-1$, which indeed gives the desired
exclusive prefix sums invariant after round $q'-1$. From this round
on, we can continue doubling as in the $1$-doubling exclusive scan
algorithm.  We term this hybrid the \emph{$q'$-doubling exclusive
scan} algorithm. It indeed contains both the $1$-doubling and the
two-$\oplus$ doubling algorithms as special cases, and is shown in
full detail as Algorithm~\ref{alg:exscan}. The algorithm doubles the
skip $s$ as long as $s<p'$ where $p'$ is chosen as $2^{q'}$ with two
applications of $\oplus$ in each round, then adjusts the skip $s$ for
the next round by subtracting one, and finishes by doubling $s$ as
long as now $s<p-1$ (since processor $r=0$ is done) with only one
application of $\oplus$ per round.  The notation $\bidirec{W,t}{T,f}$
denotes a simultaneous send and receive operation with data in buffers
$W$ and $T$ to and from processors $t$ and $f$. We use
$\shifthi{s}{1}$ to denote doubling by shifting (binary) $s$ upwards
by one position.  We have the following proposition.

\begin{proposition}
  \label{thm:exscan}
  For any chosen $q',1\leq q'\leq\ceil{\log_2 p}$, the exclusive scan
  problem is solved in $q=\ceil{\log_2
    (p-1)+\log_2\frac{2^{q'}}{2^{q'}-1}}=\ceil{\log_2\left(\frac{2^{q'}}{2^{q'}-1}(p-1)\right}$
  simultaneous send-receive communication rounds by
  Algorithm~\ref{alg:exscan} with at most $q+q'-2$ applications of
  the associative, binary operator $\oplus$ by any processor.
\end{proposition}

\begin{proof}
  Let $q'$ be the chosen number of inclusive scan rounds and take
  $p'=2^{q'}$. By the initial round and the first \textbf{while}-loop,
  there are clearly $1+(q'-1)=q'$ communication rounds for the
  inclusive scan phase. After that, processor $r=0$ is done, and the
  second \textbf{while}-loop takes $q-q'$ communication rounds for the
  exclusive scan phase. The total number of communication rounds is
  therefore $q'+(q-q')=q$. The total round count $q$ must be at least
  $\ceil{\log_2 p}$ and if chosen too large, the last rounds will, by
  the conditions, entail no communication. The exact number of actual
  communication rounds is computed as follows.
  \begin{align*}
    (2^{q'}-1)2^{q-q'} &< p-1 \Leftrightarrow \\
    2^{q-q'} &< \frac{p-1}{2^{q'}-1} \Leftrightarrow \\
    q-q' &< \log_2\frac{p-1}{2^{q'}-1} \Leftrightarrow \\
    q &< \log_2\frac{p-1}{2^{q'}-1}+q' \\
    &= \log_2\frac{p-1}{2^{q'}-1}+\log_2 2^{q'} \\
    &= \log_2(p-1) + \log_2\frac{2^{q'}}{2^{q'}-1} \\
    &= \log_2\frac{2^{q'}}{2^{q'}-1}(p-1) 
  \end{align*}
  The first round has no applications of $\oplus$, but the $q'-1$
  inclusive scan rounds have two applications each at most for any
  processor. The final $q-q'$ exclusive scan only rounds each have one
  application of $\oplus$ for the processors that are receiving
  partial results. The total number of applications is therefore at
  most $2(q'-1)+(q-q')=q+q'-2$ for any processor. 
\end{proof}

For the implementation, the smallest number of inclusive scan rounds
$q'$ for which the the total number of rounds is still at most
$\ceil{\log_2 p}$ can easily be computed.

\begin{corollary}
  \label{cor:bestq}
  The smallest number $q'$ of inclusive scan rounds in
  Algorithm~\ref{alg:exscan} for which the total number of rounds is
  $q=\ceil{\log_2 p}$ is $q'\geq q-\log_2(2^q-p+1)$.
\end{corollary}
\begin{proof}
  We solve
  \begin{align*}
    2^{q-q'}(2^{q'}-1) &\geq p-1 \Leftrightarrow \\
    2^q-2^{q-q'} &\geq p-1 \Leftrightarrow \\
    2^q-p+1 &\geq \frac{2^q}{2^{q'}} \Leftrightarrow \\
    2^{q'} &\geq \frac{2^q}{2^q-p+1}
  \end{align*}
  and and arrive at the claim by taking the logarithm.
\end{proof}

The corresponding best $p'$ to use in Algorithm~\ref{alg:exscan}
to achieve this smallest number of $q'$ inclusive scan rounds is
$2^{\ceil{\log_2 p}}-p+1$. We term the corresponding algorithm the
\emph{best-doubling exclusive scan} algorithm. 

\begin{table}
  \caption{Communication rounds $q$ versus number of $\oplus$ operator
    applications for Algorithm~\ref{alg:exscan} and
    Algorithm~\ref{alg:adjustedexscan} for small processor counts $p$,
    $24\leq p<37$ and $1140\leq p<1153$. For
    Algorithm~\ref{alg:exscan}, we show the cases $q'=1$,
    corresponding to the $1$-doubling algorithm, $q'=q$, corresponding
    to the two-$\oplus$ doubling algorithm; we list the computed $q'$
    for the best-doubling algorithm. Algorithm~\ref{alg:adjustedexscan}
    is the roughly halving algorithm.}
  \label{tab:roundsvsapps}
  \begin{center}
  \begin{tabular}{r|rr|rrr|rr|rr|rr}
    \multicolumn{1}{c}{} &
    \multicolumn{2}{c}{$q'=1$} &
    \multicolumn{3}{c}{Best $q'$} &
    \multicolumn{2}{c}{$q'=q$} &
    \multicolumn{2}{c}{Roughly halving}
    \\
  $p$ & $q$ & $\oplus$ & $q$ & $\oplus$ & $q'$ & $q$ & $\oplus$ & $q$ & $\oplus$ \\
\toprule
24 & 6 & 5 & 
5 & 5 & 
2 & 
5 & 6 & 
5 & 6
\\
25 & 6 & 5 & 
5 & 5 & 
2 & 
5 & 6 & 
5 & 5
\\
26 & 6 & 5 & 
5 & 6 & 
3 & 
5 & 6 & 
5 & 5
\\
27 & 6 & 5 & 
5 & 6 & 
3 & 
5 & 6 & 
5 & 6
\\
28 & 6 & 5 & 
5 & 6 & 
3 & 
5 & 6 & 
5 & 6
\\
29 & 6 & 5 & 
5 & 6 & 
3 & 
5 & 6 & 
5 & 6
\\
30 & 6 & 5 & 
5 & 6 & 
4 & 
5 & 6 & 
5 & 6
\\
31 & 6 & 5 & 
5 & 6 & 
4 & 
5 & 6 & 
5 & 7
\\
32 & 6 & 5 & 
5 & 6 & 
5 & 
5 & 6 & 
5 & 7
\\
33 & 6 & 5 & 
6 & 5 & 
1 & 
6 & 8 & 
6 & 5
\\
34 & 7 & 6 & 
6 & 6 & 
2 & 
6 & 8 & 
6 & 5
\\
35 & 7 & 6 & 
6 & 6 & 
2 & 
6 & 8 & 
6 & 6
\\
36 & 7 & 6 & 
6 & 6 & 
2 & 
6 & 8 & 
6 & 6
\\
\midrule
1140 & 12 & 11 & 
11 & 11 & 
2 & 
11 & 18 & 
11 & 14
\\
1141 & 12 & 11 & 
11 & 11 & 
2 & 
11 & 18 & 
11 & 14
\\
1142 & 12 & 11 & 
11 & 11 & 
2 & 
11 & 18 & 
11 & 14
\\
1143 & 12 & 11 & 
11 & 11 & 
2 & 
11 & 18 & 
11 & 15
\\
1144 & 12 & 11 & 
11 & 11 & 
2 & 
11 & 18 & 
11 & 15
\\
1145 & 12 & 11 & 
11 & 11 & 
2 & 
11 & 18 & 
11 & 14
\\
1146 & 12 & 11 & 
11 & 11 & 
2 & 
11 & 18 & 
11 & 14
\\
1147 & 12 & 11 & 
11 & 11 & 
2 & 
11 & 18 & 
11 & 15
\\
1148 & 12 & 11 & 
11 & 11 & 
2 & 
11 & 18 & 
11 & 15
\\
1149 & 12 & 11 & 
11 & 11 & 
2 & 
11 & 18 & 
11 & 15
\\
1150 & 12 & 11 & 
11 & 11 & 
2 & 
11 & 18 & 
11 & 15
\\
1151 & 12 & 11 & 
11 & 11 & 
2 & 
11 & 18 & 
11 & 16
\\
1152 & 12 & 11 & 
11 & 11 & 
2 & 
11 & 18 & 
11 & 16
\\
\bottomrule
\end{tabular}
\end{center}
\end{table}

With $q'=1$ and $q'=q$, Algorithm~\ref{alg:exscan} contains the
$1$-doubling and two-$\oplus$ doubling exclusive scan algorithms as
special cases. Choosing $p'$ as in Corollary~\ref{cor:bestq} gives the
best-doubling algorithm. The number of communication rounds $q$ and
number of $\oplus$ applications for some small number of processors
$p$ (as used in the experimental evaluation in the appendix) are
listed in Table~\ref{tab:roundsvsapps}. The chosen $q'$ for the
best-doubling algorithm is also listed. The $1$-doubling algorithm
achieves the smallest number of $\oplus$ applications at the expense
of one extra communication round. The number of $\oplus$ applications
are bounded as by Proposition~\ref{thm:exscan}. The best $q'$
increases gradually towards $q$ as $p$ approaches a power of two $2^q$
from below.

In the appendix, the three variants are ran with two concrete MPI
libraries and illustrate that the differences can have practical impact.

\section{Another Exclusive Scan Algorithm}
\label{sec:differentalgorithm}

It is known that the so-called all-reduce operation (in MPI: \mpiallreduce),
in which all processors compute the same result
\begin{align*}
  W &= \oplus_{i=0}^{p-1}V_i
\end{align*}
can be performed in exactly $\ceil{\log_2 p}$ communication rounds if
the $\oplus$operator is commutative~
\cite{BarNoyKipnisSchieber93,Traff25:reducescatter}. This is achieved
with a so-called \emph{roughly halving} circulant graph communication
pattern where the $k$th skip $s_k$ for round $k, k=0,1,\ldots q$ is
computed from the $(k+1)$th skip by halving and rounding up,
$s_{k}=\ceil{s_{k+1}/2}$ (with $s_q=p$).  The invariant maintained by
the all-reduce algorithm of~\cite{Traff25:reducescatter} can naturally
be adopted to the exclusive scan problem, even without relying on
commutativity.  The adopted invariant states that processor $r$ has
computed the partial result
\begin{align*}
  W_r & =  \oplus_{i=\max(0,r-s_{k+1}+1)}^{r-1}V_i
\end{align*}
after round $k$. In other words, processor $r$ has computed an
exclusive prefix sum over the inputs from processor $r-s_{k+1}+1$ up
to but \emph{excluding} processor $r$ itself. As for the all-reduce
algorithm~\cite{Traff25:reducescatter}, we can maintain this invariant
with a roughly halving circulant graph by distinguishing the two cases
$s_k+s_k=s_{k+1}$ ($s_{k+1}$ even) and $s_k+s_k=s_{k+1}+1$ ($s_{k+1}$
odd), and adjusting the communication pattern accordingly. If $s_{k+1}$ is
even, we receive a partial result from processor $r-s_k$, and if
$s_{k+1}$ is odd from processor $r-(s_k-1)=r-s_k+1$.

Computing $s_k$ by roughly halving from $p$, it obviously holds that
\begin{displaymath}
  s_{k+1}-1 =
  \begin{cases}
    (s_k-1)+(s_k-1) & \text{if $s_{k+1}$ is odd} \\
    (s_k-1)+(s_k-1)+1 & \text{if $s_{k+1}$ is even}
  \end{cases}
\end{displaymath}
It follows that $s_{k+1}-1$ is exactly the most significant $q-(k+1)$ bits of
$p-1=s_q-1$. We can therefore compute each $s_k$ directly
\begin{align*}
  s_k &\equiv \shiftlo{(p-1)}{(q-k)}+1 \quad 
\end{align*}
where $\shiftlo{(p-1)}{(q-k)}$ denotes shifting (binary) $p-1$
downwards by $q-k$ positions.

\begin{algorithm}
  \caption{Algorithm for the exclusive scan operation for processor
    $r$ on the circulant graph with adjusted roughly halving skips
    $s_k=\ceil{s_{k+1}/2}$ for $k=0,1,\ldots,q-1$ and $s_q=p$.  Each
    processor has an input vector $V$ and computes the result into
    $W$. The associative operator for pairwise reduction of vectors is
    $\oplus$. The first communication round has no application of
    $\oplus$ and the remaining rounds at most one or two according to
    whether $s_{k+1}$ is even or odd.}
  \label{alg:adjustedexscan}
  \begin{algorithmic}
    \Procedure{ExScan}{$V,W,\oplus$}
    \If{$p=1$} \Comment For good measure, handle $p=1$ case
    \State $W\gets V$
    \State\Return
    \EndIf \Comment Now $p>1$
    \State $f,t\gets r-s_0,r+s_0$
    \Comment To- and from-processors for round $0$
    \If{$0\leq f\band t<p$}
    \State $\bidirec{V,t,C^s_p}{W,f}$
    \ElsIf{$0\leq f$} $\recvdata(W,f)$
    \ElsIf{$t<p$} $\senddata(V,t)$
    \EndIf
    \Comment First communication round, $p>1$
    \For{$k=1,\ldots,q-1$}
    \If{$\isodd(s_{k+1})$} $\varepsilon\gets 1$
    \Else \ $\varepsilon\gets 0$
    \EndIf
    \State $f,t\gets r-s_k+\varepsilon,r+s_k-\varepsilon$
    \Comment To- and from-processors
    \If{$\varepsilon\leq f\band t<p$}
    \If{$\varepsilon=0$} $W'\gets W\oplus V$ \Comment Compute partial result
    \Else\ $W'\equiv W$ \Comment Alias to invariant result
    \EndIf
    \State $\bidirec{W',t}{T,f}$
    \State $W\gets T\oplus W$
    \ElsIf{$\varepsilon\leq f$}
    \State $\recvdata(T,f)$
    \State $W\gets T\oplus W$
    \ElsIf{$0<r\band t<p$}
    \If{$\varepsilon=0$} $W'\gets W\oplus V$ \Comment At most once, since nothing more is received
    \Else\ $W'\equiv W$
    \EndIf
    \State $\senddata(W',t)$
    \ElsIf{$r=0\band\varepsilon=0$} $\senddata(V,t)$
    \EndIf
    \EndFor
    \EndProcedure
  \end{algorithmic}
\end{algorithm}

\begin{proposition}
  \label{thm:exscanadjusted}
  The exclusive scan problem is solved in $\ceil{\log_2 p}$
  send-receive communication rounds by
  Algorithm~\ref{alg:adjustedexscan} using a roughly halving skip
  sequence $s_k=\ceil{s_{k+1}/2}=\shiftlo{(p-1)}{(q-k)}+1$ with at
  most $q+\popcount(p-1)-2-\isevn(p)$ applications of the associative
  $\oplus$ operator for any processor.
\end{proposition}

\begin{proof}
  By distinguishing the two cases $s_{k+1}$ odd and $s_{k+1}$ even, it
  can easily be seen that the claimed invariant holds. This implies
  that for each processor, $W$ is the exclusive prefix sum over the
  inputs from the smaller numbered processors.  In the first round,
  there is no application of the $\oplus$ operator by any processor.
  In the $q-1$ following commutation rounds, all receiving processors
  perform a reduction, and in the rounds where $s_{k+1}$ is even,
  sending processors perform an additional reduction. In the last
  round, if $p$ is even, all processors either send or receive, but
  not both, which yields one less reduction.  The claim then follows
  from the observation since the number of even $s_k$ for
  $k=2,\ldots,q-1$ is $\popcount(p-1)-1$.
\end{proof}

The number of $\oplus$ applications for different, small numbers of
processors $p$ are also shown in Table~\ref{tab:roundsvsapps} and are
in accordance with Proposition~\ref{thm:exscanadjusted}.

\section{Summary}

This note surveyed known algorithms for the message-passing inclusive
and exclusive scan operations (like \mpiscan and \mpiexscan in
MPI~\cite{MPI-4.1}) suitable for small input and output vectors. We
gave two efficient algorithms for the exclusive scan operation that
can both readily be used to implement the \mpiexscan collective.  The
first algorithm is a configurable, hybrid algorithm that computes the
exclusive scan in $q=\ceil{\log_2
  (p-1)+\log_2\frac{2^{q'}}{2^{q'}-1}}$ simultaneous send-receive
communication rounds with at most $q-q'-2$ applications of the given,
associative (and possibly expensive) binary operator for any given
choice of initial rounds $q'$. We also gave a best $q'$ for
$q=\ceil{\log_2 p}$ communication rounds overall.  The second
algorithm is a specialization of a round-optimal all-reduce algorithm
which runs in $\ceil{\log_2 p}$ communication rounds with
$q+\popcount(p-1)-2-\isevn(p)$ applications of the operator.
It is an open question whether an algorithm exists that
can compute the exclusive scan in $\ceil{\log_2 (p-1)}$ communication
rounds and at most the same number of operator applications. The
algorithms discussed use a circulant graph communication pattern, and
complement the algorithms discussed
in~\cite{Traff25:optimalbroadcast,Traff25:reducescatter} for all the
other MPI collective operations.

\bibliographystyle{plain}
\bibliography{traff,parallel} 

\appendix
\section{An Experimental Evaluation}
\label{sec:evaluation}

We have implemented the two exclusive scan algorithms for and in
MPI~\cite{MPI-4.1} exactly as shown in Algorithm~\ref{alg:exscan} and
Algorithm~\ref{alg:adjustedexscan}.  Simultaneous send and receive is
implemented with \mpisendrecv, and local application of a pre- or
user-defined MPI operator with \mpireducelocal. The latter is a
two-argument operation, taking an input and an input-output vector and
reducing these together in this order. In the inclusive scan phase of
Algorithm~\ref{alg:exscan} where an inclusive prefix sum partial
result $W\oplus V$ has to be sent, as well as for
Algorithm~\ref{alg:adjustedexscan} when $s_{k+1}$ is even, a
three-argument local reduction function would have been
desirable~\cite{Traff23:attributes} to avoid an extra copy into an
intermediate reduce buffer.

We aim to estimate how Algorithm~\ref{alg:exscan} performs for
different choices of $p'=2^{q'}$ in comparison to the library native
\mpiexscan operation, and how this compares against
Algorithm~\ref{alg:adjustedexscan}.  To possibly see the small
differences by one communication round, we experiment with small
$p=36$ as well as with the full system with $p=36\times 32$
processors, see Table~\ref{tab:roundsvsapps}. 

Our experimental system is a medium sized $36\times 32$ processor
cluster with $36$~dual socket compute nodes, each with two Intel(R)
Xeon(R) Gold 6130F $16$-core CPUs. The compute nodes are
interconnected via dual Intel Omnipath interconnects each with a
bandwidth of $100$ Gbytes/s.  We give results with two available MPI
libraries, namely \hydrampich and \hydraopenmpi, and observe that MPI
libraries differ highly in both absolute performance and in
scalability characteristics.  The implementations and benchmarks were
compiled with \gcc with the \texttt{-O3} option.

As element type for the scan operations, we have used \mpilong, and
\mpibxor as the binary operator. Our benchmarking procedure performs,
for each element count, \benchrep measurement repetitions with
\benchwarm warmup measurements. We synchronize the MPI processes with
\mpibarrier (twice), and for each experiment determine the time for
the slowest process to complete the exclusive scan operation. Over the
\benchrep repetitions, the minimum of these times is
reported~\cite{Traff12:mpibenchmark}.

The results for the two MPI libraries, \hydrampich and \hydraopenmpi,
for element counts from $c=0$ to $c=100\,000$ in $p=36\times 1$ and
$p=36\times 32$ MPI processes are shown in
Table~\ref{tab:sometimes-mpich} and Table~\ref{tab:sometimes-openmpi}.

\begin{table}
  \caption{Measured running times for the native \mpiexscan operation
    and the $q'$-doubling algorithm with the \hydrampich library in the
    $p=36\times 1$ and $p=36\times 32$ MPI process configurations,
    respectively, for element counts $c=1,10,100,1\,000,10\,000,100\,000$ and
    $q'=1,q$, best $q'$ and the roughly halving halving.}
  \label{tab:sometimes-mpich}
  \begin{center}
    \begin{tabular}{rrrrrr}
      & \multicolumn{5}{c}{$p=36\times 1$ MPI processes} \\
      \midrule
      $c$ &
      \mpiexscan & $1$-doubling & best-doubling & two-$\oplus$ & halving \\
      (\mpilong) & ($\mu$seconds) & ($\mu$seconds) & ($\mu$seconds) & ($\mu$seconds) & ($\mu$seconds) \\
      \toprule
      0 & 
      0.08 &
      0.06 &
      0.05 &
      0.05 &
      0.07 \\
      1 & 
      9.85 &
      9.18 &
      9.12 &
      8.20 &
      10.40 \\
      10 &
      17.22 &
      18.22 &
      16.83 &
      15.91 &
      17.38 \\
      100 & 
      19.05 &
      19.70 &
      18.08 &
      17.26 &
      19.34 \\
      1000 & 
      36.39 &
      34.66 &
      32.47 &
      35.10 &
      36.00 \\
      10000 & 
      255.52 &
      254.67 &
      230.51 &
      267.71 &
      248.16 \\
      100000 & 
      2774.59 &
      1571.19 &
      1510.11 &
      1969.52 &
      1707.75 \\      
      \bottomrule
      & \multicolumn{5}{c}{$p=36\times 32$ MPI processes} \\
      \midrule
      $c$ &
      \mpiexscan & $1$-doubling & best-doubling & two-$\oplus$ & halving \\
      (\mpilong) & ($\mu$seconds) & ($\mu$seconds) & ($\mu$seconds) & ($\mu$seconds) & ($\mu$seconds) \\
      \toprule
      0 & 
      0.09 &
      0.07 &
      0.07 &
      0.07 &
      0.07 \\
      1 & 
      27.53 &
      37.10 &
      27.08 &
      22.20 &
      27.69 \\
      10 & 
      32.64 &
      41.01 &
      35.86 &
      33.99 &
      37.22 \\
      100 & 
      37.76 &
      41.54 &
      40.53 &
      39.08 &
      42.48 \\
      1000 & 
      164.14 &
      156.18 &
      146.94 &
      158.86 &
      157.13 \\
      10000 & 
      1015.59 &
      1095.81 &
      1037.51 &
      1113.78 &
      1106.17 \\
      100000 & 
      11959.38 &
      11187.50 &
      10997.74 &
      14938.65 &
      13319.05 \\
      \bottomrule
    \end{tabular}
  \end{center}
\end{table}

\begin{table}
  \caption{Measured running times for the native \mpiexscan operation
    and the $q'$-doubling algorithm with the \hydraopenmpi library in the
    $p=36\times 1$ and $p=36\times 32$ MPI process configurations,
    respectively, for element counts $c=1,10,100,1\,000,10\,000,100\,000$ and
    $q'=1,q$, best $q'$ and the roughly halving halving.}
  \label{tab:sometimes-openmpi}
  \begin{center}
    \begin{tabular}{rrrrrr}
      & \multicolumn{5}{c}{$p=36\times 1$ MPI processes} \\
      \midrule
      $c$ &
      \mpiexscan & $1$-doubling & best-doubling & two-$\oplus$ & halving \\
      (\mpilong) & ($\mu$seconds) & ($\mu$seconds) & ($\mu$seconds) & ($\mu$seconds) & ($\mu$seconds) \\
      \toprule
      0 &
0.06 &
0.05 &
0.05 &
0.06 &
0.07 \\
      1 &
44.63 &
9.15 &
7.58 &
8.07 &
8.61 \\
      10 &
54.15 &
11.24 &
9.43 &
10.06 &
10.65 \\
      100 &
64.08 &
13.60 &
11.04 &
12.13 &
12.53 \\
      1000 &
163.16 &
31.56 &
26.64 &
32.65 &
32.05 \\
      10000 &
1127.28 &
286.57 &
261.02 &
290.96 &
294.75 \\
      100000 &
7379.84 &
1844.17 &
1729.44 &
2246.06 &
2078.20 \\
      \bottomrule
      & \multicolumn{5}{c}{$p=36\times 32$ MPI processes} \\
      \midrule
      $c$ &
      \mpiexscan & $1$-doubling & best-doubling & two-$\oplus$ & halving \\
      (\mpilong) & ($\mu$seconds) & ($\mu$seconds) & ($\mu$seconds) & ($\mu$seconds) & ($\mu$seconds) \\
      \toprule
      0 &
0.18 &
0.06 &
0.15 &
0.17 &
0.40 \\
1 &
917.53 &
15.56 &
16.10 &
14.12 &
17.06 \\
      10 &
1462.47 &
19.26 &
19.13 &
17.41 &
21.73 \\
100 &
2017.98 &
29.10 &
26.80 &
25.41 &
30.38 \\
      1000 &
6618.10 &
135.03 &
129.20 &
137.45 &
143.70 \\
      10000 &
34326.68 &
1065.49 &
1000.15 &
1030.18 &
1097.96 \\
      100000 &
357144.19 &
10423.22 &
10674.83 &
14456.33 &
12879.54 \\      
      \bottomrule
    \end{tabular}
  \end{center}
\end{table}

The \hydraopenmpi library obviously has a poor implementation of
\mpiexscan upon which the new algorithms both improve easily by a
large factor.  The \hydrampich library (which seems the overall best
performing library among our available libraries) apparently applies a
logarithmic round algorithm all the way up to the large count of
$c=100\,000$ elements; this can possibly be improved by using other
algorithms~\cite{LakhotiaPetriniKannanPrasanna22,Traff09:twotree,Traff06:scan}.
Neither of the libraries seem to do the \mpiexscan operation in a
hierarchical fashion when the compute nodes are full; this can
possibly be improved~\cite{Traff20:mpidecomp}.  The new algorithms
perform comparably in the two libraries: The \hydraopenmpi library has
better performance in comparison when the compute nodes are full,
whereas the \hydrampich library performs better in the case with only
one MPI process per node.

For both libraries, the best-doubling algorithm almost everywhere
fares the best. For the small element counts with $c=1,10,100$, the
two-$\oplus$ algorithm can sometimes be slightly better (with very
little difference to the $q-1$-doubling algorithm), but as $c$ grows,
the extra applications of $\oplus$ noticeably decreases the
performance.

The results illustrate that care must be taken when implementing
\mpiexscan, and that the configurability offered by
Algorithm~\ref{alg:exscan} can have a practical impact. With the best
choice of $q'$, this algorithm also does better than
Algorithm~\ref{alg:adjustedexscan}. The implementation of both
algorithms is straightforward, and it would make sense to include
either in MPI libraries.

\end{document}